\documentclass{article}
\usepackage{amsfonts, amsmath, amsthm, latexsym, amscd, graphicx, xypic, psfrag}
\usepackage{epic, eepic, epsfig, setspace, ifthen, amssymb, xy, yhmath, lscape}
\usepackage{bbold, amscd, paralist, graphics}
\usepackage{scalerel,stackengine}
\usepackage{thmtools}
\usepackage[colorlinks=true]{hyperref}

\def\br#1{\left(#1\right)}
\def\Br#1{\left[#1\right]}
\def\BR#1{\left\{#1\right\}}

\def\<#1,#2>{\langle #1,#2 \rangle}

\def\bothID{\rlap{\hbox to.97\wd0{\hss\vrule height.06\ht0 width.82\wd0}}
\copy0\rlap{\kern-.36\wd0\vrule height1.05\ht0 width.05\ht0}\kern.14\wd0}

\DeclareMathOperator{\dom}{dom}

\DeclareMathOperator{\Var}{Var}

\newtheorem{theorem}{Theorem}
\newtheorem{defi}[theorem]{Definition}

\newtheorem{corollary}[theorem]{Corollary}

\newtheorem{definition}[theorem]{Definition}

\newtheorem{lemma}[theorem]{Lemma}

\newtheorem{proposition}[theorem]{Proposition}
\declaretheoremstyle[bodyfont=\normalfont]{normalbody}
\declaretheorem[style=normalbody, sharenumber=theorem]{remark}
\declaretheorem[style=normalbody, sharenumber=theorem]{example}

\stackMath
\newcommand\reallywidehat[1]{%
\savestack{\tmpbox}{\stretchto{%
  \scaleto{%
    \scalerel*[\widthof{\ensuremath{#1}}]{\kern-.6pt\bigwedge\kern-.6pt}%
    {\rule[-\textheight/2]{1ex}{\textheight}}
  }{\textheight}%
}{0.5ex}}%
\stackon[1pt]{#1}{\tmpbox}%
}
\def\br#1{\left(#1\right)}
\def\Br#1{\left[#1\right]}
\def\BR#1{\left\{#1\right\}}

\DeclareMathOperator{\diag}{diag} 

\DeclareMathOperator{\cov}{Cov}

\begin{document}

\title{When Risks and Uncertainties Collide: Quantum Mechanical Formulation of Mathematical Finance for Arbitrage Markets}

\author{Simone Farinelli\\
        Core Dynamics GmbH\\
        Scheuchzerstrasse 43\\
        CH-8006 Zurich\\
        Email: simone@coredynamics.ch\\and\\
        Hideyuki Takada\\
        Department of Information Science\\
        Narashino Campus, Toho University\\
        2-2-1-Miyama, Funabashi-Shi\\ J-274-8510 Chiba\\
        Email: hideyuki.takada@is.sci.toho-u.ac.jp
        }

\maketitle

\begin{abstract}
Geometric arbitrage theory reformulates a generic asset model possibly allowing for arbitrage by packaging all asset and their forward dynamics into a stochastic principal fibre bundle, with a connection whose parallel transport encodes discounting and portfolio rebalancing, and whose curvature measures, in this geometric language, the ''instantaneous arbitrage capability'' generated by the market itself.
The asset and market portfolio dynamics have a quantum mechanical description, which is constructed by quantizing the deterministic version of the stochastic Lagrangian system describing a market allowing for arbitrage.
Results, obtained by solving the Schr\"odinger equations, coincide with those obtained by solving the stochastic Euler Lagrange equations derived by a variational principle and providing therefore consistency.
\end{abstract}

\tableofcontents

\section{Introduction}
This paper further develops a conceptual structure - called geometric arbitrage theory - to
link arbitrage modeling in generic markets with quantum mechanics.\par
Geometric arbitrage theory rephrases classical stochastic finance in stochastic differential geometric
terms in order to characterize arbitrage.
 The main idea of the geometric arbitrage theory approach
consists of modeling markets made of basic financial instruments
together with their term structures as principal fibre bundles.
Financial features of this market - like no arbitrage and
equilibrium - are then characterized in terms of standard
differential geometric constructions - like curvature - associated
to a natural connection in this fibre bundle.
Principal fibre bundle theory has been heavily exploited in
theoretical physics as the language in which laws of nature can be
best formulated by providing an invariant framework to describe
physical systems and their dynamics. These ideas can be carried
over to mathematical finance and economics. A market is a
financial-economic system that can be described by an appropriate
principle fibre bundle. A principle like the invariance of market
laws under change of num\'{e}raire can be seen then as gauge
invariance. Concepts like No-Free-Lunch-with-Vanishing-Risk (NFLVR) and No-Unbounded-Profit-with-Bounded-Risk (NUPBR) have a geometric characterization, which have
the Capital Asset Pricing Model (CAPM) as a consequence.\par
The fact that gauge theories are the natural language
to describe economics was first proposed by Malaney and Weinstein
in the context of the economic index problem (\cite{Ma96},
\cite{We06}). Ilinski (see \cite{Il00} and \cite{Il01}) and Young
(\cite{Yo99}) proposed to view arbitrage as the curvature of a
gauge connection, in analogy to some physical theories.
Independently, Cliff and Speed (\cite{SmSp98}) further developed
Flesaker and Hughston seminal work (\cite{FlHu96}) and utilized
techniques from differential geometry (indirectly mentioned by
allusive wording) to reduce the complexity of asset models before
stochastic modelling.\par 
The contribution of Cliff \& Speed are independent from those of Ilinski, Malaney and Young. The
former represent a base financial asset as an ordered couple of a deflator and a term structure, where the deflator is the asset value expressed in term of some num\'{e}raire and the term structure models the price structure of futures, i.e. linear derivatives of the base assets. Gauge transforms correspond to the linear portfolio construction and are utilized to represent a financial market by means of a minimal set of gauges on which stochastic modelling is to be applied. But Cliff \& Speed do not use the formalism of differential stochastic geometry. Ilinski, Malaney and Young, on the contrary, after having introduced a different principal bundle structure corresponding to dilations, utilize a connection to define the parallel transport. Ilinski, moreover, is the first one to consider the concept of curvature applied to financial markets. \par
This paper is structured
as follows. Section 2 reviews
classical stochastic finance and  geometric arbitrage theory. Arbitrage is seen as
curvature of a principal fibre bundle representing the market which defines the
quantity of arbitrage associated to it. Proofs are omitted and can be found in \cite{Fa15}, \cite{Fa20} and in \cite{FaTa20}, where geometric arbitrage theory has been
given a rigorous mathematical foundation utilizing the background
of stochastic differential geometry as in Schwartz (\cite{Schw80}),
Elworthy (\cite{El82}), Em\'{e}ry(\cite{Em89}), Hackenbroch \& Thalmaier (\cite{HaTh94}),
Stroock (\cite{St00}) and Hsu (\cite{Hs02}).\par
Section 3 describes the intertwined dynamics of assets, term structures and market portfolio as constrained Lagrange system deriving it from a stochastic variational principle whose Lagrange function measures the arbitrage quantity allowed by the market.
This constrained Lagrange system and its stochastic Euler-Lagrange equation is equivalent to a constrained Hamilton system, obtained by Legendre transform, with its stochastic Hamilton equations.
These stochastic Hamilton system is, on its turn, equivalent to a quantum mechanical system, obtained by quantizing the deterministic version of the Hamilton system.
This is shown in Section 4 where we reformulate mathematical finance in terms of quantum mechanics. The Schr\"odinger equation describes then both the asset and market portfolio dynamics, which can be explicitly computed once the spectrum of the Hamilton operator is known. Without knowledge of the spectrum it is still possible by means of Ehrenfest's theorem to determine stochastic properties of  future asset values and market portfolio nominals.  Their expected values  are identical to those computed with the stochastic Euler Lagrange equation, demonstrating the consistency of the quantum mechanical approach. Moreover, we prove that for a closed market, the returns on market portfolio nominals, asset values and term structures are centered and serially uncorrelated. Hence, the justification of econometrics with its autocorrelated models or those with stochastic volatilities lies in the fact that markets are not closed: there is always an asset category which has not been modeled, to which wealth can escape, destroying the uncorrelated identical distributional behaviour of the remaining asset categories.
By applying Heisenberg's uncertainty relation to the quantum mechanical model of the market we obtain one more econometric result: the volatilities of asset values on one hand and their weights in the market portfolio are mutually exclusive, meaning by this that if the former increase, the latter decreases and vice-versa.
\par In Section 5 we solve the Schr\"odinger equation representing the arbitrage market dynamics by using Feynman's path integrals.
Appendix A reviews and generalizes Nelson's stochastic derivatives. Section 6 concludes.

\section{Geometric Arbitrage Theory Background}\label{section2}
In this section we explain the main concepts of geometric arbitrage theory introduced
in \cite{Fa15}, to which we refer for proofs and examples.
\subsection{The Classical Market Model}\label{StochasticPrelude}
In this subsection we will summarize the classical set up, which
will be rephrased in section (\ref{foundations}) in differential
geometric terms. We basically follow \cite{HuKe04} and \cite{DeSc08}.\par We assume continuous time trading and
that the set of trading dates is $[0,+\infty[$. This assumption is
general enough to embed the cases of finite and infinite discrete
times as well as the one with a finite horizon in continuous time.
 This motivates the
technical effort of continuous time stochastic finance.\par The  outcome  of  chance is modelled by a filtered probability space
$(\Omega,\mathcal{A}, \mathbb{P})$, where $\mathbb{P}$ is the
statistical (physical) probability measure,
$\mathcal{A}=\{\mathcal{A}_t\}_{t\in[0,+\infty[}$ an increasing
family of sub-$\sigma$-algebras of $\mathcal{A}_{\infty}$ and
$(\Omega,\mathcal{A}_{\infty}, \mathbb{P})$ is a probability space.
The filtration $\mathcal{A}$ is assumed to satisfy the usual
conditions: $\mathcal{A}_t=\bigcap_{s>t}\mathcal{A}_s$ for all $t\in[0,+\infty[$ (right continuity),
and $\mathcal{A}_0$ contains all null sets of
$\mathcal{A}_{\infty}$.

The market consists of finitely many \textbf{assets} indexed by
$j=1,\dots,N$, whose \textbf{nominal prices} are given by the
vector valued semimartingale $S:[0,+\infty[\times\Omega\rightarrow\mathbb{R}^N$
denoted by $(S_t)_{t\in[0,+\infty[}$ adapted to the filtration $\mathcal{A}$.
The stochastic process $(S^ j_t)_{t\in[0,+\infty[}$ describes the
price at time $t$ of the $j$th asset in terms of  unit of cash
\textit{at time $t=0$}. More precisely, we assume the existence of a
$0$th asset, the \textbf{cash}, a strictly positive
semimartingale, which evolves according to
$S_t^0=\exp(\int_0^tdu\,r^0_u)$, where the integrable
semimartingale $(r^0_t)_{t\in[0,+\infty[}$ represents the
continuous interest rate provided by the cash account: one always
knows in advance what the interest rate on the own bank account
is, but this can change from time to time. The cash account is
therefore considered the locally risk less asset in contrast to
the other assets, the risky ones. In the following we will mainly
utilize \textbf{discounted prices}, defined as
$\hat{S}_t^j:=S_t^j/S^{0}_t$, representing the asset prices in
terms of \textit{current} unit of cash.\par
 We remark that there is no need to
assume that asset prices are positive. But, there must be at least
one strictly positive asset, in our case the cash. If we want to
renormalize the prices by choosing another asset instead of the
cash as reference, i.e. by making it to our
\textbf{num\'{e}raire}, then this asset must have a strictly
positive price process. More precisely, a generic num\'{e}raire is
an asset, whose nominal price is represented by a strictly
positive stochastic process $(B_t)_{t\in[0,+\infty[}$, and
 which is a portfolio of the original assets $j=0,1,2,\dots,N$. The discounted prices of the original
assets are  then represented in terms of the num\'{e}raire by the
semimartingales $\hat{S}_t^j:=S_t^j/B_t$.\par We assume that there
are no transaction costs and that short sales are allowed. Remark
that the absence of transaction costs can be a serious limitation
for a realistic model. The filtration $\mathcal{A}$ is not
necessarily generated by the price process
$(S_t)_{t\in[0,+\infty[}$: other sources of information than
prices are allowed. All agents have access to the same information
structure, that is to the filtration $\mathcal{A}$.\par
Let $v\ge0$. A  $v$-admissible \textbf{strategy} $x=(x_t)_{t\in[0,+\infty[}$ is a predictable semimartingale for which the It\^{o} integral $\int_0^tx\cdot dS\ge-v$ for all $t\ge0$. A strategy is admissible if it is $v$-admissible for some $v\ge0$.
An admissible strategy $x$ is said to be \textbf{self-financing}
if and only if $V_t:=x_t\cdot S_t$, the portfolio value at time $t$, is given by
\begin{equation}
V_t=V_0+\int_0^tx_u\cdot dS_u.
\end{equation}

\begin{defi}[\textbf{Arbitrage}] Let the process $(S_t)_{[0,+\infty[}$ be a semimartingale and $(x_t)_{t\in[0,+\infty[}$ be admissible self-financing strategy. Let us consider trading up to time $T\le\infty$. The portfolio wealth at time $t$ is given by by $V_{t}(x):=V_0+\int_0^tx_u\cdot dS_u$, and we denote by $K_0$ the subset of $L^0(\Omega, \mathcal{A}_{T},P)$ containing all such $V_T(x)$, where $x$ is any admissible self-financing strategy.
We define:
\begin{enumerate}
\item $C_0:=K_0-L_+^0(\Omega, \mathcal{A}_{T},P)$.
\item $C:=C_0\cap L_+^{\infty}(\Omega, \mathcal{A}_{T},P)$.
\item $\bar{C}$: the closure of $C$ in $L^{\infty}$ with respect to the norm topology.
\item $\mathcal{V}^{V_0}:=\left\{(V_{t})_{t\in[0,+\infty[}\,\big{|}\, V_t=V_t(x), \,\text{where } x \text{ is } V_0\text{-admissible} \right\}$.
\item $\mathcal{V}_T^{V_0}:=\left\{V_T\,\big{|}\,(V_{t})_{t\in[0,+\infty[}\in\mathcal{V}^{V_0}\right\}$: terminal wealth for $V_0$-admissible self-financing strategies.
\end{enumerate}
We say that $S$ satisfies
\begin{enumerate}
\item \textbf{(NA), no arbitrage}, if and only if $C \cap L^{\infty}(\Omega, \mathcal{A}_{T},P)=\{0\}$.
\item \textbf{(NFLVR), no-free-lunch-with-vanishing-risk},  if and only if $\bar{C} \cap L^{\infty}(\Omega, \mathcal{A}_{T},P)=\{0\}$.
\item \textbf{(NUPBR), no-unbounded-profit-with-bounded-risk}, if and only if $\mathcal{V}_T^{V_0}$ is bounded in $L^0$ for some $V_0>0$.
\end{enumerate}
\end{defi}
\noindent The relationship between these three different types of arbitrage has been elucidated in \cite{DeSc94} and in \cite{Ka97} with the proof of the following result.
\begin{theorem}
\begin{equation}
\text{(NFLVR)}\Leftrightarrow \text{(NA)}+\text{(NUPBR)}.
\end{equation}
\end{theorem}
\begin{remark} We recall that, as shown in \cite{DeSc94, Ka97}, (NUPBR) is equivalent to ($NAA_1$), i.e.
no asymptotic arbitrage of the $1$st kind , and equivalent to ($NA_1$), i.e. no arbitrage of the 1st kind.
\end{remark}
\begin{theorem}[\textbf{First fundamental theorem of asset pricing}]\label{FFTAP}
The market $(S,\mathcal{A})$ satisfies the (NFLVR) condition if and only if there exists an equivalent local martingale measure $P^*$.
\end{theorem}
\begin{remark}
In the first fundamental theorem of asset pricing we just assumed that the price process $S$ is locally bounded. If $S$ is bounded, then (NFLVR) is equivalent to the existence of a martingale measure.
But without this additional assumption (NFLVR) only implies the existence of a \textit{local} martingale measure, i.e. a local martingale which is \textit{not} a martingale. This distinction is important, because the difference between a security price process being a strict local martingale  versus a martingale under a probability $P^*$ relates to the existence of asset price bubbles.
\end{remark}

\subsection{Geometric Reformulation of the Market Model: Primitives}
We are going to introduce a more general representation of the
market model introduced in Subsection \ref{StochasticPrelude}, which
better suits to the arbitrage modeling task.
\begin{defi}\label{defi1}
A \textbf{gauge} is an ordered pair of two $\mathcal{A}$-adapted
real valued semimartingales $(D, P)$, where
$D=(D_t)_{t\ge0}:[0,+\infty[\times\Omega\rightarrow\mathbb{R}$ is
called \textbf{deflator} and
$P=(P_{t,s})_{t,s}:\mathcal{T}\times\Omega\rightarrow\mathbb{R}$,
which is called \textbf{term structure}, is considered as a
stochastic process with respect to the time $t$, termed
\textbf{valuation date} and
$\mathcal{T}:=\{(t,s)\in[0,+\infty[^2\,|\,s\ge t\}$. The parameter
$s\ge t$ is referred as \textbf{maturity date}. The following
properties must be satisfied a.s. for all $t, s$ such that $s\ge
t\ge 0$:
 \begin{itemize}
  \item [(i)] $P_{t,s}>0$,
  \item [(ii)] $P_{t,t}=1$.
 \end{itemize}
\end{defi}

\begin{remark}
Deflators and term structures can be considered \textit{outside the context of fixed income.} An arbitrary financial instrument is mapped to a gauge $(D, P)$ with the following economic interpretation:
\begin{itemize}
\item Deflator: $D_t$ is the value of the financial instrument at time $t$ expressed in terms of some num\'{e}raire. If we choose the cash account, the $0$-th asset as num\'{e}raire, then we can set $D_t^j:=\hat{S}_t^j=\frac{S_t^j}{S_t^0}\quad(j=1,\dots N)$.
\item Term structure: $P_{t,s}$ is the value at time $t$ (expressed in units of deflator at time $t$) of a synthetic zero coupon bond with maturity $s$ delivering one unit of financial instrument at time $s$. It represents a term structure of forward prices with respect to the chosen num\'{e}raire.
\end{itemize}
\noindent We point out that there is no unique choice for deflators and term structures describing an asset model. For example, if a set of deflators qualifies, then we can multiply every deflator  by the same positive semimartingale to obtain another suitable set of deflators. Of course term structures have to be modified accordingly. The term ``deflator" is clearly inspired by actuarial mathematics and  was first introduced by Smith \& Speed in \cite{SmSp98}. In the present context it refers to an asset value up division by a strictly positive semimartingale (which can be the state price deflator if this exists and it is made to the num\'{e}raire). There is no need to assume that a deflator is a positive process. However, if we want to make an asset to our num\'{e}raire, then we have to make sure that the corresponding deflator is a strictly positive stochastic process.
\end{remark}

\subsection{Geometric Reformulation of the Market Model: Portfolios}\label{trans}
We want now to introduce transforms of deflators and term structures
in order to group gauges containing the same (or less) stochastic
information. That for, we will consider \textit{deterministic}
linear combinations of assets modelled by the same gauge (e. g. zero
bonds of the same credit quality with different maturities).

\begin{defi}\label{gaugeTransforms2}
Let $\pi:[0, +\infty[\longrightarrow \mathbb{R}$ be a deterministic
cashflow intensity (possibly generalized) function. It induces a
\textbf{gauge transform} $(D,P)\mapsto
\pi(D,P):=(D,P)^{\pi}:=(D^{\pi}, P^{\pi})$ by the formulae
\begin{equation}
 D_t^{\pi}:=D_t\int_0^{+\infty}dh\,\pi_h P_{t, t+h}\qquad
P_{t,s}^{\pi}:=\frac{\int_0^{+\infty}dh\,\pi_h P_{t,
s+h}}{\int_0^{+\infty}dh\,\pi_h P_{t, t+h}}.
\end{equation}

\end{defi}
\begin{remark} The cashflow intensity $\pi$ specifies the bond cashflow
structure. The bond value at time $t$ expressed in terms of the
market model  num\'{e}raire is given by $D_t^{\pi}$.  The term
structure of forward prices for the bond future expressed in terms
of the bond current value is given by $P_{t,s}^{\pi}$.
\end{remark}

\begin{proposition}\label{conv}
Gauge transforms induced by cashflow vectors have the following
property:
\begin{equation}((D,P)^{\pi})^{\nu}= ((D,P)^{\nu})^{\pi} = (D,P)^{\pi\ast\nu},\end{equation} where
$\ast$ denotes the convolution product of two cashflow vectors or
intensities respectively:
\begin{equation}\label{convdef}
    (\pi\ast\nu)_t:=\int_0^tdh\,\pi_h\nu_{t-h}.
\end{equation}

\end{proposition}

\noindent The convolution of two non-invertible gauge transform is
non-invertible. The convolution of a non-invertible with an
invertible gauge transform is non-invertible.


\begin{defi}\label{int}
If the term structure is differentiable with respect to the maturity date, it can be written as a functional of the
\textbf{instantaneous forward rate } f defined as
\begin{equation}
  f_{t,s}:=-\frac{\partial}{\partial s}\log P_{t,s},\quad
  P_{t,s}=\exp\left(-\int_t^sdhf_{t,h}\right).
\end{equation}
\noindent and
\begin{equation}
 r_t:=\lim_{s\rightarrow t^+}f_{t,s}
\end{equation}
\noindent is termed \textbf{short rate}.
\end{defi}

\begin{remark} The special choice of vanishing interest rate $r\equiv0$ or flat term structure
$P\equiv1$ for all assets corresponds to the classical model,
where only asset prices and their dynamics are relevant.
\end{remark}

\subsection{Arbitrage Theory in a Differential Geometric
Framework}\label{foundations} Now we are in the position to
rephrase the asset model presented in Subsection
\ref{StochasticPrelude} in terms of a natural geometric language.
Given $N$ base assets we want to construct a
portfolio theory and study arbitrage and thus we cannot a priori assume the existence of a
risk neutral measure or of a state price deflator. The market model is seen as a principal
fibre bundle of the (deflator, term structure) pairs, discounting
and foreign exchange as a parallel transport, num\'{e}raire as
global section of the gauge bundle, arbitrage as curvature.  The
no-free-lunch-with-vanishing-risk condition is proved to be
equivalent to a zero curvature condition.

\subsubsection{Market Model as Principal Fibre Bundle}
Let us consider -in continuous time- a market with $N$ assets and a
num\'{e}raire. A general portfolio at time $t$ is described by the
vector of nominals $x\in \mathfrak{X}$, for an open set
$\mathfrak{X}\subset\mathbb{R}^N$. Following Definition \ref{defi1}, the asset model consisting in $N$ synthetic zero bonds is described by means of the gauges
\begin{equation}(D^j,P^j)=((D_t^j)_{t\in[0, +\infty[},(P_{t,s}^j)_{s\ge t}),\end{equation}
\noindent where $D^j$ denotes the deflator and $P^j$ the term
structure for $j=1,\dots,N$. This can be written as
\begin{equation}P_{t,s}^j=\exp\left(-\int_t^sf^j_{t,u}du\right),\end{equation}
where $f^j$ is the instantaneous forward rate process for the $j$-th asset and the corresponding short rate is given by $r_t^j:=\lim_{u\rightarrow 0^+}f^j_{t,u}$. For a
portfolio with nominals $x\in \mathfrak{X}\subset\mathbb{R}^N$ we define
\begin{equation}
D_t^x:=\sum_{j=1}^Nx_jD_t^j\qquad
f_{t,u}^x:=\sum_{j=1}^N\frac{x_jD_t^j}{\sum_{j=1}^Nx_jD_t^j}f_{t,u}^j\qquad
P_{t,s}^x:=\exp\left(-\int_t^sf^x_{t,u}du\right).
\end{equation}
The short rate writes
\begin{equation}
r_t^x:=\lim_{u\rightarrow 0^+}f^x_{t,u}=\sum_{j=1}^N\frac{x_jD_t^j}{\sum_{j=1}^Nx_jD_t^j}r_t^j.
\end{equation}
The image space of all possible strategies reads
\begin{equation}M:=\{(t,x)\in [0,+\infty[\times\mathfrak{X}\}.\end{equation}
In Subsection \ref{trans} cashflow intensities and the corresponding
gauge transforms were introduced. They have the structure of an
Abelian semigroup
\begin{equation}
 H:=\mathcal{E}^{\prime}([0,
+\infty[,\mathbb{R})=\{F\in\mathcal{D}^{\prime}([0,+\infty[)\mid
\text{supp}(F)\subset[0, +\infty[\text{ is compact}\},
\end{equation}
where the semigroup operation on distributions with compact support
is the convolution (see \cite{Ho03}, Chapter IV), which extends the
convolution of regular functions as defined by formula
(\ref{convdef}).
\begin{defi}\label{MFB}
The \textbf{Market Fibre Bundle} is defined as the fibre bundle of
gauges
\begin{equation}
\mathcal{B}:=\{ ({D^x_t},{P^x_{t,\,\cdot}})^{\pi }|\,(t,x)\in
M, \pi\in G\}.
\end{equation}
\end{defi}\noindent
The cashflow intensities defining invertible transforms constitute
an Abelian group
\begin{equation}
G:=\{\pi\in H|\text{ it exists } \nu\in H\text{ such that
}\pi\ast\nu=\delta\}\subset \mathcal{E}^{\prime}([0,
+\infty[,\mathbb{R}).
\end{equation}
From Proposition \ref{conv} we obtain
\begin{theorem} The market fibre bundle $\mathcal{B}$ has the
structure of a $G$-principal fibre bundle  given by the action
\begin{equation}
\begin{split}
\mathcal{B}\times G &\longrightarrow\mathcal{B}\\
 ((D,P), \pi)&\mapsto (D,P)^{\pi}=(D^{\pi},P^{\pi})
\end{split}
\end{equation}
\noindent The group $G$ acts freely and differentiably on
$\mathcal{B}$ to the right.
\end{theorem}

\subsubsection{Stochastic Parallel Transport} Let us
consider the projection of $\mathcal{B}$ onto $M$
\begin{equation}
\begin{split}
p:\mathcal{B}\cong M\times G&\longrightarrow M\\
 (t,x,g)&\mapsto (t,x)
\end{split}
\end{equation}
and its differential map at $(t,x,g) \in \mathcal{B}$ denoted by $T_{(t,x,g)}p$, see for example, Definition 0.2.5 in (\cite{Bl81}) 
\begin{equation}
T_{(t,x,g)}p:\underbrace{T_{(t,x,g)}\mathcal{B}}_{\cong\mathbb{R}^N\times\mathbb{R}\times
\mathbb{R}^{[0, +\infty[}}\longrightarrow \ \
\underbrace{T_{(t,x)}M}_{\cong\mathbb{R}^N\times\mathbb{R}}.
\end{equation}
The vertical directions are
\begin{equation}\mathcal{V}_{(t,x,g)}\mathcal{B}:=\ker\left(T_{(t,x,g)}p\right)\cong\mathbb{R}^{[0, +\infty[},\end{equation}
and the horizontal ones are
\begin{equation}\mathcal{H}_{(t,x,g)}\mathcal{B}\cong\mathbb{R}^{N+1}.\end{equation}
An Ehresmann connection on $\mathcal{B}$ is a projection
$T\mathcal{B}\rightarrow\mathcal{V}\mathcal{B}$. More precisely, the
vertical projection must have the form
\begin{equation}
\begin{split}
\Pi^v_{(t,x,g)}:T_{(t,x,g)}\mathcal{B}&\longrightarrow\mathcal{V}_{(t,x,g)}\mathcal{B}\\
(\delta x,\delta t, \delta g)&\mapsto (0,0,\delta g +
\Gamma(t,x,g).(\delta x, \delta t)),
\end{split}
\end{equation}
and the horizontal one must read
\begin{equation}
\begin{split}
\Pi^h_{(t,x,g)}:T_{(t,x,g)}\mathcal{B}&\longrightarrow\mathcal{H}_{(t,x,g)}\mathcal{B}\\
(\delta x,\delta t, \delta g)&\mapsto (\delta x,\delta
t,-\Gamma(t,x,g).(\delta x, \delta t)),
\end{split}
\end{equation}
such that
\begin{equation}\Pi^v+\Pi^h=\mathbb{1}_{\mathcal{B}}.\end{equation}
Stochastic parallel transport on a principal fibre bundle along a
semimartingale is a well defined construction (cf. \cite{HaTh94},
Chapter 7.4 and \cite{Hs02} Chapter 2.3 for the frame bundle case)
in terms of Stratonovich integral. Existence and uniqueness can be
proved analogously to the deterministic case by formally
substituting the deterministic time derivative $\frac{d}{dt}$ with
the stochastic one $\mathcal{D}$ corresponding to the Stratonovich
integral.\par
 Following Ilinski's idea (\cite{Il01}), we motivate the choice
of a particular connection by the fact that it allows to encode
foreign exchange and discounting as parallel transport.
\begin{theorem}\label{Ilinski}With the choice of connection
\begin{equation}\label{connection}\chi(t,x,g).(\delta x, \delta t):= \left(\frac{D_t^{\delta x}}{D_t^x}-r_t^x\delta t\right) g,\end{equation}
the parallel transport in $\mathcal{B}$ has the following financial
interpretations:
\begin{itemize}
\item Parallel transport along the nominal directions ($x$-lines)
corresponds to a multiplication by an exchange rate.
\item Parallel transport along the time direction ($t$-line)
corresponds to a division by a stochastic discount factor.
\end{itemize}
\end{theorem}

Recall that time derivatives needed to define the parallel transport
along the time lines have to be understood in Stratonovich's sense.
We see that the bundle is trivial, because it has a global
trivialization, but the connection is not trivial.

\subsubsection{Nelson $\mathcal{D}$ Weak Differentiable Market Model} We continue to reformulate the classic asset model introduced in Subsection \ref{StochasticPrelude} in terms of stochastic differential geometry.
\begin{definition}\label{weakMM}
 A \textbf{Nelson $\mathcal{D}$ weak differentiable market model} for $N$ assets is described by $N$ gauges which are Nelson $\mathcal{D}$ weak differentiable with respect to the time variable. More exactly, for all $t\in[0,+\infty[$ and $s\ge t$ there is an open time interval $I\ni t$ such that for the deflators $D_t:=[D_t^1,\dots,D_t^N]^{\dagger}$ and the term structures $P_{t,s}:=[P_{t,s}^1,\dots,P_{t,s}^N]^{\dagger}$, the latter seen as processes in $t$ and parameter $s$, there exist a $\mathcal{D}$ weak $t$-derivative (see Appendix \ref{Derivatives}). The short rates are defined by $r_t:=\lim_{s\rightarrow t^{-}}\frac{\partial}{\partial s}\log P_{ts}$.\par
 A strategy is a curve $\gamma:I\rightarrow X$ in the portfolio space parameterized by the time. This means that the allocation at time $t$ is given by the vector of nominals $x_t:=\gamma(t)$. We denote by $\bar{\gamma}$ the lift of $\gamma$ to $M$, that is $\bar{\gamma}(t):=(\gamma(t),t)$. A strategy is said to be \textbf{closed} if it represented by a closed curve.  A \textbf{weak $\mathcal{D}$-admissible strategy} is predictable and $\mathcal{D}$- weak differentiable.
\end{definition}
\begin{remark}
We require weak $\mathcal{D}$-differentiability and not strong $\mathcal{D}$-differentiability because imposing a priori regularity properties on the trading strategies corresponds to restricting the class of admissible strategies with respect to the classical notion of Delbaen
and Schachermayer. Every (no-)arbitrage consideration depends crucially on the chosen definition
of admissibility. Therefore, restricting the class of admissible strategies may lead to
the automatic exclusion of potential arbitrage opportunities, leading to vacuous statements of FTAP-like results. An admissibile strategy in the classic sense (see Section \ref{section2}) is weak $\mathcal{D}$-differentiable.
\end{remark}
\noindent In general the allocation can depend on the state of the nature i.e. $x_t=x_t(\omega)$ for $\omega\in\Omega$.
\begin{proposition}
A weak $\mathcal{D}$-admissible strategy is self-financing if and only if
\begin{equation}\label{sf}
\mathcal{D}(x_t\cdot D_t)=x_t\cdot \mathcal{D}D_t-\frac{1}{2}\mathfrak{D}_*\left<x,D\right>_t\text{ or }
\mathcal{D}x_t\cdot D_t=-\frac{1}{2}\mathfrak{D}_*\left<x,D\right>_t\text{ or }
\mathfrak{D}x_t\cdot D_t=0,
\end{equation}
almost surely. The bracket $\left<\cdot,\cdot\right>$ denotes the continuous part of the quadratic covariation.
\end{proposition}

For the remainder of this paper unless otherwise stated we will deal
only with $\mathcal{D}$ differentiable market models, $\mathcal{D}$
differentiable strategies, and, when necessary, with $\mathcal{D}$
differentiable state price deflators. All It\^{o} processes are
$\mathcal{D}$ differentiable, so that the class of considered
admissible strategies is very large.

\subsubsection{Arbitrage as Curvature}
 The Lie algebra of $G$ is
\begin{equation}\mathfrak{g}=\mathbb{R}^{[0, +\infty[}\end{equation}
and therefore commutative. The $\mathfrak{g}$-valued curvature $2$-form is defined by means the $\mathfrak{g}$-valued connection
$1$-form as
\begin{equation}R:=d\chi+[\chi,\chi],\end{equation} meaning by this,
that for all $(t,x,g)\in \mathcal{B}$ and for all $\xi,\eta\in
T_{(t,x)}M$
\begin{equation}R(t,x,g)(\xi,\eta):=d\chi(t,x,g)(\xi,\eta)+[\chi(t,x,g)(\xi),\chi(t,x,g)(\eta)]=d\chi(t,x,g)(\xi,\eta). \end{equation}
Remark that, being the Lie algebra commutative, the Lie bracket
$[\cdot,\cdot]$ vanishes.
\begin{proposition}[\textbf{Curvature Formula}]\label{curvature}
Let $R$ be the curvature. Then, the following quality holds:
\begin{equation}R(t,x,g)=g dt\wedge d_x\left[\mathcal{D} \log (D_t^x)+r_t^x\right].\end{equation}
\end{proposition}
\noindent The following result  characterizes arbitrage as curvature.
\begin{theorem}[\textbf{No Arbitrage}]\label{holonomy}
The following assertions are equivalent:
\begin{itemize}
\item [(i)] The market model (with base assets and futures with discounted prices $D$  and $P$) satisfies the no-free-lunch-with-vanishing-risk condition.
\item[(ii)] There exists a positive semimartingale $\beta=(\beta_t)_{t\ge0}$ such that deflators and short rates satisfy for all portfolio nominals and all times the condition
\begin{equation}r_t^x=-\mathcal{D}\log(\beta_tD_t^x).\end{equation}
\item[(iii)] There exists a positive semimartingale $\beta=(\beta_t)_{t\ge0}$ such that deflators and term structures satisfy for all portfolio nominals and all times the condition
\begin{equation}P^x_{t,s}=\frac{\mathbb{E}_t[\beta_sD^x_s]}{\beta_tD^x_t}.\end{equation}
\end{itemize}
\end{theorem}

\noindent This motivates the following definition.
\begin{defi}
The market model satisfies the \textbf{zero curvature (ZC)} if and only if
the curvature vanishes a.s.
\end{defi}

\noindent Therefore, we have following implications relying  two different definitions of no-abitrage:
\begin{corollary}
\begin{equation}
\text{(NFLVR)}\Rightarrow \text{(ZC)}.
\end{equation}
\end{corollary}

\begin{remark}
The positive semimartingale $\beta=(\beta_t)_{t\ge0}$ in Theorem \ref{holonomy} is termed pricing kernel or state price deflator throughout the literature.
\end{remark}

\section{Asset and Market Portfolio Dynamics as a Constrained Lagrangian System}\label{Hamilton}
In \cite{Fa15} and \cite{Fa20} the minimal arbitrage principle, stating that asset dynamics and market portfolio choose the path
guaranteeing the minimization of arbitrage, was encoded as the Hamilton principle under constraints for a Lagrangian measuring the arbitrage.
Then, the SDE describing asset deflators, term structures and market portfolio were derived by means of a stochasticization procedure
of the Euler-Lagrange equations following a technique developed by Cresson and Darses (\cite{CrDa07} who follow previous works of Yasue (\cite{Ya81}) and Nelson (\cite{Ne01}).
Since we need this set up to proceed with its quantization, we briefly summarize it here below.
\begin{definition}
Let $\gamma$ be the market $\mathcal{D}$-admissible strategy,
 and $\delta\gamma$, $\delta D$,
$\delta r$ be perturbations of the market strategy,
deflators' and short rates' dynamics. The \textbf{variation} of
$(\gamma,D,r)$ with respect to the given perturbations is the
following one parameter family:
\begin{equation}\epsilon\longmapsto(\gamma^\epsilon,D^\epsilon,r^\epsilon):=(\gamma,D,r)+ \epsilon(\delta \gamma,\delta D,\delta r).\end{equation}
Thereby, the parameter $\epsilon$ belongs to some open neighborhood
of $0\in\mathbb{R}$. The \textbf{arbitrage action} with respect to a positive semimartingale $\beta$ can be
consistently defined by
\begin{equation}\label{int_arb}
\begin{split}
A^{\beta}(\gamma;D,r)&:=\int_{\gamma}dt\BR{\mathcal{D}\log(\beta_tD_t^{x_t})+r_t^{x_t}}=\\
&=\int_0^Tdt\frac{x_t\cdot (\mathcal{D}D_t+r_tD_t)}{x_t\cdot D_t}+\log\frac{\beta_1}{\beta_0},
\end{split}
\end{equation}where $x=x_t$ is an admissible self-financing strategy taking values on the curve $\gamma$, and the first
variation of the arbitrage action as
\begin{equation}\delta A^{\beta}(\gamma;D,r):=\frac{d}{d\epsilon}A^{\beta}(\gamma^{\epsilon};D^{\epsilon},r^{\epsilon})\mid_{\epsilon:=0}.\end{equation}
\end{definition}
\noindent This leads to the following
\begin{definition}
Let us introduce the notation $q:=(x,D,r)$ and
$q^{\prime}:=(x^{\prime},D^{\prime},r^{\prime})$ for two vectors in
$\mathbb{R}^{3N}$. The \textbf{Lagrangian (or Lagrange function)} is
defined as
\begin{equation}\label{Lagrangian}
L(q,q^{\prime}):=L(x,D,r,x^{\prime},D^{\prime},r^{\prime}):=\frac{x\cdot (D^{\prime}+ rD)}{x\cdot D}.
\end{equation}
The self-financing \textbf{constraint} is defined as
\begin{equation}
C(q,q^{\prime}):=x^{\prime}\cdot D.
\end{equation}
\end{definition}

\begin{lemma}
The arbitrage action for a self-financing
strategy $\gamma$ is the integral of the Lagrange function along the
$\mathcal{D}$-admissible strategy:
\begin{equation}
A^{\beta}(\gamma;D,r)=\int_{\gamma}dt\,L(q_t,q_t^{\prime})+\log\frac{\beta_1}{\beta_0}=\int_{\gamma}dt\,L(x_t,D_t,r_t,x_t^{\prime},D_t^{\prime},r_t^{\prime})+\log\frac{\beta_1}{\beta_0}.
\end{equation}
\end{lemma}

\noindent A fundamental result of classical mechanics allows to
compute the extrema of the arbitrage action in the
\textit{deterministic} case as the solution of a system of ordinary
differential equations.

\begin{theorem}[\textbf{Hamilton Principle}]\label{ham}
Let us denote the derivative with respect to time as
$\frac{d}{dt}=:\prime$ and assume that all quantities observed are
deterministic. The local extrema of the arbitrage action satisfy the
Lagrange equations under the self-financing constraints
\begin{equation}\label{hamilton_principle}
\left\{
\begin{array}{l}
\delta A^{\beta}(\gamma;D,r)=0\text{ for all } (\delta\gamma,\delta D, \delta r)\\
\text{such that }\; {x_t^{\prime}}^{\epsilon}\cdot D^{\epsilon }_t=0 \text{ for all }\epsilon
\end{array}
\right.
\Longleftrightarrow
\left\{
       \begin{array}{l}
        \frac{d}{dt}\frac{\partial L_{\lambda}}{\partial q^{\prime}}-\frac{\partial L_{\lambda}}{\partial q}=0\\
        C(q,q^{\prime}):=x^{\prime}\cdot D=0
       \end{array}
\right.
\end{equation}
\noindent where $\lambda\in\mathbb{R}$ denotes the the self-financing constraint Lagrange multiplier and $L_{\lambda}:=L-\lambda C$.
\end{theorem}
\noindent Let $L=L(q,q^{\prime})$ be the Lagrange function of a deterministic
Lagrangian system with the non holonomic constraint $C(q,q^{\prime})=0$. Setting $L_{\lambda}:=L-\lambda C$ for the constraint Lagrange multipliers the dynamics is given by the extended Euler-Lagrange
equations
\begin{equation}
\text{(EL)}\quad\left\{
         \begin{array}{l}
          \frac{d}{dt}\frac{\partial L_{\lambda}}{\partial q^\prime}(q,q^{\prime})-\frac{\partial L_{\lambda}}{\partial q}(q,q^{\prime})=0\\
          C(q,q^{\prime})=0
         \end{array}
       \right.
\end{equation}
meaning by this that the deterministic solution $q=q_t$ and $\lambda\in\mathbb{R}$ satisfy the constraint and
\begin{equation}\frac{d}{dt}\frac{\partial L_{\lambda}}{\partial q^\prime}\br{q_t, \frac{dq_t}{dt}}-\frac{\partial L_{\lambda}}{\partial
q}\br{q_t, \frac{dq_t}{dt}}=0.\end{equation}
\begin{defi}
 The formal
\textbf{stochastic embedding of the Euler-Lagrange equations} is
obtained by the formal substitution
\begin{equation}S:\frac{d}{dt}\longmapsto\mathcal{D},\end{equation}
and allowing the coordinates of the tangent bundle to be stochastic
\begin{equation}
\text{(SEL)}\quad\left\{
  \begin{array}{l}
    \mathcal{D}\frac{\partial L_{\lambda}}{\partial q^\prime}(q,q^{\prime})-\frac{\partial L_{\lambda}}{\partial q}(q,q^{\prime})=0\\
    C(q,q^{\prime})=0
  \end{array}
\right.
\end{equation}
meaning by this that the stochastic solution $Q=Q_t$ and the random variable $\lambda$ satisfy the constraint and
\begin{equation}\label{SEL}
\left\{
  \begin{array}{l}
  \mathcal{D}\frac{\partial L_{\lambda}}{\partial q^\prime}\br{Q_t, \mathcal{D}Q_t}-\frac{\partial L_{\lambda}}{\partial q}\br{Q_t, \mathcal{D}Q_t}=0\\
  C(Q_t,\mathcal{D}Q_t)=0.
  \end{array}
\right.
\end{equation}
\end{defi}
\noindent Let $L=L(q,q^{\prime})$ be the Lagrange function of a deterministic
Lagrangian system on a time interval $I$ with constraint $C=0$. Set
\begin{equation}\Xi:=\BR{Q\in\mathcal{C}^1(I)\mid\mathbb{E}\Br{\int_I|L_{\lambda}(Q_t,\mathcal{D}Q_t)|dt}<+\infty}.\end{equation}
\begin{defi}The action functional associated to $L_{\lambda}$ defined by
\begin{equation}
\begin{aligned}
&F: &\Xi\longrightarrow &\mathbb{R}\\
& &Q\longmapsto&\mathbb{E}\Br{\int_IL_{\lambda}(Q_t,\mathcal{D}Q_t)dt}
\end{aligned}
\end{equation}
 is called \textbf{stochastic analogue of the classic action} under the constraint $C=0$.
\end{defi}
For a sufficiently smooth extended Lagrangian $L_{\lambda}$ a necessary and sufficient
condition for a stochastic process to be a critical point of the
action functional $F$ is the fulfillment of the stochastic
Euler-Lagrange equations (SEL), as it can be seen in Theorem 7.1 page 54 in
\cite{CrDa07}. Moreover we have the following
\begin{lemma}[\textbf{Coherence}]
The following diagram commutes
\begin{equation}
    \xymatrix{
        L_{\lambda}(q_t,q_t^{\prime}) \ar[r]^{S} \ar[d]_{\text{Critical Action Principle}} & L_{\lambda}(Q_t,\mathcal{D}Q_t) \ar[d]^{\text{Stochastic Critical Action Principle}} \\
        (EL) \ar[r]_{S}       & (SEL) }
\end{equation}
\end{lemma}

\section{Asset and Market Portfolio Dynamics as Solution of the Schr\"odinger Equation: The Quantization of the Deterministic Constrained Hamiltonian System}
There are two ways of obtaining a stochastic theory starting from a deterministic one.
The first way is the one adopted by Cresson \& Darses \cite{CrDa07}, which analyzes a deterministic Lagrangian system or an equivalent deterministic Hamilton system and studies what happens if we make all variables  stochastic variables. The biggest difficult is the time $t$-dependence, which cannot be defined pathwise and is related to the definitions of It\^{o}'s integrals or, equivalenty, to Stratonovich's integrals. Since for the Stratonovich's time derivative corresponding to the Stratonovich's integral we have the same chain rule for the composition of two functions, It\^{o}'s Lemma for the Stratonovich derivative has the same form as the chain rule for deterministic functions. This is the main ingredient for the Cresson-Darses stochasticization construction. Their conclusion is that, in order to solve the stochastic Euler-Lagrange equations or their equivalent, the stochastic Hamilton equations, it suffices to solve their deterministic counterpart and add a stochastic perturbation with zero expectation, satisfying several constraints. This is the way summarized in Section \ref{Hamilton} and has been carried out in \cite{Fa20}.\par

The second way to obtain a stochastic theory from a deterministic one, relies on quantization of a deterministic system. In classical physics one describes the time evolution of the system by either Newton's second law of the dynamics, or the equivalent deterministic Euler-Lagrange equations, or the equivalent Hamilton equations. In the Euler-Lagrange equations the Lagrange function describing the system appears, while in the Hamilton equations, the Hamilton function appears. The Hamilton function is obtained by Legendre transform of the Lagrange function. The quantization step consists in introducing an Hilbert Space of $L^2$ complex valued functions over the coordinate space for the deterministic (real) variables (all but the time), in our case $x$, $D$ and $r$: this coordinate space is a manifold, and the Hilbert space is a $L^2$ space over this manifold.  We consider the unit ball of this Hilbert space whose elements are defined as the states of the quantum mechanical system: the interpretation of  the square of the absolute value of those complex valued functions with $L^2$ norm one is that of a probability density. To obtain the time dynamics we have to construct the Hamilton operator derived by the Hamilton function. This procedure is called quantization of the Hamilton function and, according on the type of the Hamilton function, it is not always a well defined procedure. In our case we are lucky. By replacing the coordinate functions of the manifold by the corresponding multiplication operators with those coordinates and by replacing the tangential plane coordinates by the partial derivatives with respect to the corresponding coordinates in the manifold, we obtain an operator, which can be symmetrized, and later made to a selfadjoint operator, by an appropriate choice of ist domain of definition: the Hamilton operator. Given the Hamilton operator we can solve the Schr\"odinger equation, which gives the time evolution of the quantum mechanical initial state state as a function of time.\par 
Finally we have two approaches who should be equivalent, that is leading to the same solution. The asset dynamics and the  market portfolio dynamics. With the first approach we obtain the stochastic processes ($ x_t,D_t,r_t)$ as a solution of the stochastic Euler-Lagrange equations. With the second approach we obtain a $L^2$ function $\psi_t(x,D,r)$ such that $|\psi_t(x,D,r)|^2$ is the probability density of the random vector $(x_t, D_t, r_t)$.\par
We proceed now by introducing an equivalent quantum mechanical representation of the asset and market portfolio dynamics. As a general background
to the mathematics of quantum mechanics we refer to \cite{Ta08} and \cite{Ha13}.
\begin{proposition}
The Hamilton function $H$ defined as Legendre transform of the Lagrangian $L$ is
\begin{equation}\label{hamilton_function}
H(p,q):=\left.\left(p\cdot q^{\prime}-L_{\lambda}(q,q^{\prime}))\right)\right|_{p:=\frac{\partial L}{\partial q^{\prime}}}= -\frac{x\cdot (rD)}{x\cdot D},
\end{equation}
where $q=(x,D,r,\lambda)$ and $p:=\frac{\partial L}{\partial q^{\prime}}=(p_x,p_D,p_r,p_{\lambda})$. Thereby, following \cite{Kl01} we have elevated the Lagrange multiplier  corresponding to the self financing constraint $C$ to an additional dynamic variable $\lambda$ with its conjugate momentum $p_\lambda$. The Hessian matrix of $L_{\lambda}$ is singular, which translates into the open first class additional constraints
$p_r=0$ and $p_{\lambda}=0$.
\end{proposition}

\begin{proof}
It follows directly by inserting
\begin{equation}\label{p_def}
\begin{split}
&p_x=\frac{\partial L_{\lambda}}{\partial x^{\prime}}=-\lambda D\\
&p_D=\frac{\partial L_{\lambda}}{\partial D^{\prime}}=\frac{x}{x\cdot D}\\
&p_r=\frac{\partial L_{\lambda}}{\partial r^{\prime}}=0\\
&p_\lambda=\frac{\partial L_{\lambda}}{\partial \lambda^{\prime}}=0
\end{split}
\end{equation}
into equations (\ref{Lagrangian}) and (\ref{hamilton_principle}).\\The last two equations are constraints on the conjugate momenta. They are first class because we have the following equalities among Poisson brackets for all $i,j$
\begin{equation}
\left\{
\begin{split}
&\{p_r^i,p_r^j\}=0\\
&\{p_\lambda,p_\lambda\}=0\\
&\{p_r^i,p_\lambda\}=0\\
&\{p_r^i,H\}=\frac{x^iD^i}{x\cdot D}\\
&\{p_\lambda^i,H\}=0.
\end{split}
\right.
\end{equation}
\end{proof}

\begin{proposition}
The selfadjoint Hamilton operator obtained by the standard quantization procedure
\begin{equation}
\begin{split}
q&\longrightarrow q\text{ (multiplication operator)}\\
p&\longrightarrow \frac{1}{\imath}\frac{\partial}{\partial q}\text{ (differential operator)}
\end{split}
\end{equation}
\noindent is
\begin{equation}\label{hamilton_operator}
H:= -\frac{x\cdot (rD)}{x\cdot D}
\end{equation}
\noindent with domain of definition
\begin{equation}
\begin{split}
\dom(H)&:=\Big{\{}\varphi\in L^2(\mathfrak{X}\times \mathbb{R}^{2N+1},\mathbb{C},d^{3N}qd\lambda)\,\left|\,H\varphi\in L^2(\mathfrak{X}\times \mathbb{R}^{2N+1},\mathbb{C},d^{3N}qd\lambda)\right.\\
&\qquad\qquad\qquad\qquad\qquad\qquad\qquad\qquad\qquad \frac{\partial \varphi}{\partial r^i}=0\,\text{ for all }i,\,\frac{\partial \varphi}{\partial \lambda}=0\Big{\}}.
\end{split}
\end{equation}

\end{proposition}
\begin{proof}
The quantization procedure of first class constrained Hamiltonian systems, explained in \cite{Di64}, \cite{Kl01} and \cite{FaJa88}, is directly applied here.\\
\end{proof}
\begin{remark}
The Hamilton operator $H$ is a multiplication operator with a real function, and is self adjoint. By Dirac's theory of first constrained quantized systems, if the constraints are satisfied at time $t=0$, then they are automatically satisfied for all times for the solution of Schr\"odinger's equation. Since  $H$ does not explicitly depend on the Lagrange multiplier $\lambda$ as well as any $\varphi$ in its domain of definition, we can drop any reference to $\lambda$ and write
\begin{equation}
\dom(H):=\Big{\{}\varphi\in L^2(\mathfrak{X}\times \mathbb{R}^{2N},\mathbb{C},d^{3N}q)\,\left|\,H\varphi\in L^2(\mathfrak{X}\times \mathbb{R}^{2N},\mathbb{C},d^{3N}q)\right.\frac{\partial \varphi}{\partial r^i}=0\,\text{ for all }i,\,\Big{\}}.
\end{equation}

\end{remark}

\begin{theorem}
The asset and market portfolio dynamics is given by the solution of the Schr\"odinger equation
\begin{equation}
\left\{
\begin{array}{l}
\imath\frac{d}{dt}\psi(q,t)=H\psi(q,t)\\
\psi(q,0)=\psi_0(q),
\end{array}
\right.
\end{equation}
\noindent where $\psi_0$ is the initial state satisfying $C\psi_0=0$ and $\int_{\mathfrak{X}\times \mathbb{R}^{2N}}dq^{3N}|\psi_0(q)|^2=1$.\\
The solution is given by
\begin{equation}\label{solution_schroedinger}
\psi(q,t)=e^{-\imath H t}\psi_0,
\end{equation}
where $\{e^{-\imath H t}\}_{t\ge 0}$ is the strong continuous, unitary one parameter group associated to the selfadjoint $H$ by Stone's theorem.
\end{theorem}
\begin{remark}\label{mom_tech}
This is the quantum mechanical formulation of the constrained stochastic Lagrangian system described by the SDE  (\ref{SEL}).
The interpretation of $|\psi(q,t)|^2$ is the probability density at time $t$ for the coordinates $q$:
\begin{equation}
P[q_t\in \mathcal{Q}]=\int_{\mathcal{Q}}dq^{3N}|\psi(q,t)|^2.
\end{equation}
Therefore, if we have a random variable $a_t=a(p,q,t)$, by mean of its quantization
\begin{equation}
A:=a\left(\frac{1}{\imath}\frac{\partial}{\partial q},q,t\right)
\end{equation}
\noindent we can compute its expectation by means of both the \textbf{Schr\"{o}dinger  and Heisenberg representation} as
\begin{equation}
\mathbb{E}_0[a_t]=(A\psi,\psi)=\int_{\mathfrak{X}\times \mathbb{R}^{2N}}dq^{3N}A\psi(q,t)\bar{\psi}(q,t)=\int_{\mathfrak{X}\times \mathbb{R}^{2N}}dq^{3N}A_t\psi(q,0)\bar{\psi}(q,0),
\end{equation}
\noindent where the time dependent operator $A_t$, the \textbf{Heisenberg representation} of the operator $A$ is defined as
\begin{equation}
A_t:=e^{iHt}Ae^{-iHt}.
\end{equation}
\noindent Higher moments of random variables $a_t$ , like any measurable functions $f(a_t)$ of them can be computed by means of this technique as
 \begin{equation}
\mathbb{E}_0[f(a_t)]=(f(A)\psi,\psi)=\int_{\mathfrak{X}\times \mathbb{R}^{2N}}dq^{3N}f(A_t)\psi(q,0)\bar{\psi}(q,0),
\end{equation}
\noindent transforming the problem in one of operator calculus.
\end{remark}
\begin{theorem}[\textbf{Ehrenfest}]\label{Ehrenfest}
The time derivative of the expectation of a selfadjoint operator $A$ is given by
\begin{equation}\label{EH}
\frac{d}{dt}(A\psi,\psi)=\frac{1}{\imath}([A,H]\psi,\psi)+\left(\frac{\partial A}{\partial t}\psi,\psi\right).
\end{equation}
\end{theorem}

\noindent A direct consequence of Ehrenfest's result is the ``energy conservation'' theorem.
\begin{corollary}
\begin{equation}\
\frac{d}{dt}(H\psi,\psi)\equiv0.
\end{equation}
\end{corollary}

\begin{corollary}The dynamics of the expected values of market portfolio, asset values and term structures is given by
\begin{equation}\label{mean_dynamics}
\begin{split}
&\mathbb{E}_0[x_t]\equiv \text{const}\\
&\mathbb{E}_0[D_t]\equiv \text{const}\\
&\mathbb{E}_0[r_t]\equiv \text{const}.
\end{split}
\end{equation}
\end{corollary}
\begin{proof}It follows direct from Ehrenfest's Theorem \ref{Ehrenfest}, because the multiplication operators $x$, $D$ and $r$ commute with the Hamilton operator $H$.
\end{proof}

\begin{remark}
Note that formulae (\ref{mean_dynamics}) coincide with those in (\cite{Fa20}), where the stochastic Lagrange equations (\ref{SEL}) have been explicitly solved. These demonstrates the consistency and compatibility of the quantum mechanical reformulation to mathematical finance.
\end{remark}
We can now derive the consequences of Ehrenfest's theorem in the standard QM representation for the stochastic Hamiltonian/Lagrangian representation of geometric arbitrage theory. 
\begin{corollary}\label{corid} The stochastic processes $(x_t)_t$, $(D_t)_t$ and $(r_t)_t$ for market portfolio nominals, asset values and term structures  are identically distributed along time.
\end{corollary}
\begin{proof}
It suffices to apply Ehrenfest's theorem to any non negative power of the operators $x$, are independent of do not depend on time  and so must be their the distribution functions.
\end{proof}

\begin{corollary}\label{coruncorr} The stochastic processes $(\mathcal{D}x_t)_t$, $(\mathcal{D}D_t)_t$ and $(\mathcal{D}r_t)_t$ for the returns market portfolio nominals, asset values and term structures are centered and serially uncorrelated.
\end{corollary}
\begin{proof}
It suffices to prove it for the nominals, because for the asset values and term structures the proofs are formally the same. By taking the time derivative of the first equality in (\ref{mean_dynamics}) we obtain
\begin{equation}
\mathbb{E}_0\left[\mathcal{D}x_t\right]=0.
\end{equation}
Let $t_1\neq t_2$. By applying Ehrenfest's Theorem \ref{Ehrenfest} we obtain
\begin{equation}
\mathbb{E}_0\left[x_{t_1}x_{t_2}\right]\equiv \text{const},
\end{equation}
where componentwise multiplication of the vector components is meant. Hence, together with the first equality in (\ref{mean_dynamics})
\begin{equation}
\cov_0\left(x_{t_1},x_{t_2}\right)\equiv \text{const},
\end{equation}
and by differentiating with respect to time $t_1$ and time $t_2$
\begin{equation}
\cov_0\left(\mathcal{D}_{t_1}x_{t_1},\mathcal{D}_{t_2}x_{t_2}\right)\equiv 0.
\end{equation}
Since $t_1\neq t_2$ are arbitrary we conclude that all autocovariances for any non zero lag of $(\mathcal{D}x_t)_t$ vanish, meaning that the process is serially uncorrelated.
\end{proof}

\begin{theorem}[\textbf{Heisenberg's uncertainty relation}]\label{Heisenberg}
Let $A$ and $B$ two selfadjoint operators on $\mathcal{H}$. The variance of the corresponding observables in the state $\varphi\in\dom(A)\cap\dom(B)$ is
\begin{equation}
\sigma_\varphi^2(A):=\|A\phi-\|A\phi\|^2\|^2\qquad\sigma_\varphi^2(B):=\|B\phi-\|B\phi\|^2\|^2,
\end{equation}
\noindent where $\|\cdot\|$ and $(\cdot,\cdot)$ are the norm and the scalar product in $\mathcal{H}=L^2(\mathfrak{X}\times\mathbb{R}^{2N},\mathbb{C},d^{3N}q)$. Then,
\begin{equation}
\sigma_\varphi^2(A)\sigma_\varphi^2(B)\ge\frac{1}{4}\|[A,B]\varphi\|^2.
\end{equation}
\end{theorem}
\noindent The proof of Theorem \ref{Heisenberg} can be found f.i. in \cite{Ta08} or in \cite{Ha13}.
\noindent By applying Heisenberg's uncertainty relation to the quantum mechanical represenation of our market model we obtain the following
\begin{proposition}\label{prop_vols}
The dynamics of the volatilities of market portfolio and asset values satisfies the inequalities
\begin{equation}\label{vols}
\Var_0\left(D_t^j\right)\Var_0\left(\frac{x_t^j}{x_t\cdot D_t}\right)\ge\frac{1}{4},
\end{equation}
\noindent for all indices $j=1,\dots,N$.
\end{proposition}
\begin{proof}
For $q=(x,D)$ we choose $A:=q^i$ and $B:=\frac{1}{\imath}\frac{\partial}{\partial q^j}$, obtaining by Theorem \ref{Heisenberg}, since $[A,B]=\imath\delta^{i,j}$, $\varphi_t=e^{-\imath t H}\varphi_0$ and $\|\varphi_t\|^2=1$,
\begin{equation}
\sigma_{\varphi_t}^2(q^j)\sigma_{\varphi_t}^2\left(\frac{1}{\imath}\frac{\partial}{\partial q^j}\right)\ge\frac{1}{4}.
\end{equation}
\noindent Now we can identify
\begin{equation}\label{ident}
\begin{split}
&\sigma_{\varphi_t}^2(q^j)=\Var_0(q_t^j)\\
&\sigma_{\varphi_t}^2\left(\frac{1}{\imath}\frac{\partial}{\partial q^j}\right)=\Var_0(p_t^j),
\end{split}
\end{equation}
\noindent and (\ref{vols}) follows after inserting the second equations of (\ref{p_def})
\begin{equation}
p_D=\frac{\partial L}{\partial D}=\frac{x}{x\cdot D}.
\end{equation}
\noindent  The proof is completed.\\
\end{proof}
Heisenberg's uncertainty relation has an interesting econometric consequence. In Proposition \ref{vols}, where we assume that the market forces minimize arbitrage, the dynamics of the market portfolio $(x_t,D_t,r_t)$ is such that
\begin{equation}\label{weight}
\Var_0\left(D_t^j\right)\Var_0\left(\frac{w_t^j}{D_t^j}\right)\ge\frac{1}{4},
\end{equation}
where $w_t^j:=\frac{x_t^jD_t^j}{x_t\cdot D_t}$ is the market weight of the $j$-th asset. This means that the less volatile the $j$-th asset value is, the more volatile its market weight must be, so that (\ref{weight}) is fulfilled.
\begin{remark} If we try to obtain the analogon of Proposition \ref{vols} for the conjugate variables $r$ and $p_r$, we face a technical problem.
As a matter of fact we can apply Theorem \ref{Heisenberg} to $q=r$, but, since $p_r=0$ as computed in the third equation (\ref{p_def}), we cannot give the interpretation of
$\Var({p_r}_t^j)\equiv0$ to  $\sigma_{\varphi_t}^2\left(\frac{1}{\imath}\frac{\partial}{\partial r^j}\right)\neq0$, similarly to (\ref{ident}). The same difficulty occurs for $q=D$.
\end{remark}
\begin{example}
Let us consider an asset with no futures, where the asset discounted prices and the market portfolio nominals are Brownian processes
\begin{equation}\label{sys}
\begin{split}
\hat{S}_t&=\hat{S}_0+\alpha_t+\sigma_tW_t\\
x_t&=x_0+a_t+s_tW_t,
\end{split}
\end{equation}
\noindent where $(W_t)_{t\in[0,+\infty[}$ is a standard $P$-Brownian motion in $\mathbb{R}^K$, for some $K\in\mathbb{N}$,
$(\sigma_t)_{t\in[0,+\infty[}$ and $(s_t)_{t\in[0,+\infty[}$  are  $\mathbb{R}^{N\times K}$-valued  differentiable functions, and $(\alpha_t)_{t\in[0,+\infty[}$ and $(a_t)_{t\in[0,+\infty[}$ are  $\mathbb{R}^N$-valued  differentiable functions. 
The coefficients $\sigma, s,\alpha,a$ must satisfy certain conditions implied by Corollaries  \ref{corid} and \ref{coruncorr}. Since 
\begin{equation}
\begin{split}
\mathbb{E}_0[\hat{S}_t]&=\hat{S}_0+\alpha_t\quad
\Var_0(\hat{S}_t)=\diag(\sigma_t\sigma_t^{\dagger})t\\
\mathbb{E}_0[x_t]&=x_0+\alpha_t\quad
\Var_0(x_t)=\diag(s_ts_t^{\dagger})t,
\end{split}
\end{equation} 
from Corollary \ref{corid} we derive the conditions
\begin{equation}
\alpha_t\equiv0,\quad a_t\equiv0,\quad \diag(\sigma_t\sigma_t^{\dagger})=\frac{\diag(\sigma_0\sigma_0^{\dagger})}{t},\quad \diag(s_ts_t^{\dagger})=\frac{\diag(s_0s_0^{\dagger})}{t},
\end{equation}
end hence,
\begin{equation}
\hat{S}_t\sim\mathcal{N}(\hat{S}_0,\sigma_0\sigma_0^{\dagger})\quad x_t\sim\mathcal{N}(x_0,s_0s_0^{\dagger})
\end{equation}
\noindent under which Corollary \ref{coruncorr} is automatically satisfied.
\end{example}

\section{The (Numerical) Solution of the Schr\"odinger Equation via Feynman Integrals}
\subsection{From the Stochastic Euler-Lagrangian Equations to Schr\"odinger's Equation: Nelson's method}
Following chapter 14 of \cite{Ne85} we consider diffusions on $N$-dimensional Riemannian manifold satisfying the SDE
\begin{equation}\label{diff}
d\xi_t=b(t,\xi_t)dt+\sigma(\xi_t)dW_t,
\end{equation}
where $(W_t)_{t\ge0}$ is a $K$-dimensional Brownian motion, and
\begin{equation}
b:[0,+\infty[\times\mathbb{R}^N\rightarrow\mathbb{R}^N\;\text{ and }\;\sigma:\mathbb{R}^N\rightarrow\mathbb{R}^{N\times K}
\end{equation}
\noindent are vector and matrix valued functions with appropriate regularity. We assume that
\begin{equation}
\sigma^2(q)(q^{\prime},q^{\prime}):=q^{\prime}\sigma(q)\sigma^{\dagger}(q)q^{\prime}=\sum_{j=1}^N{q^{\prime}}^jq^{\prime}_j
\end{equation}
defines a Riemannian metric, and introduce the notation $v^j:=\sum_{i=1}^N(\sigma\sigma^{\dagger})_{j,i}\,v_i$.\\
We consider a Lagrangian on $M$ given as
\begin{equation}\label{Lagr}
L(q, q^{\prime}, t):=\sum_{j=1}^N\left[\frac{1}{2}{q^{\prime}}^jq^{\prime}_j-\Phi(q)+A_j(q){q^{\prime}}^j\right],
\end{equation}
for given potentials $\Phi$ and $A$.
For the diffusion (\ref{diff}) the Guerra-Morato Lagrangian writes
\begin{equation}
L_+(\zeta,t):=\sum_{j=1}^N\left[\frac{1}{2}b^j(t,\zeta)b_j(t,\zeta)+\frac{1}{2}\nabla_jb^j(t,\zeta)-\Phi(\zeta)+A_j(\zeta)b^j(t,\zeta)+\frac{1}{2}\nabla_jA^j(\zeta)\right]
\end{equation}
We define
\begin{equation}
R(t,q):=\frac{1}{2}\log\rho(t,q),
\end{equation}
where $\rho$ is the density of the process $(\xi_t)_{t\ge0}$, and
\begin{equation}
S(t,q):=\mathbb{E}\left[\left.\int_0^tL_+(\xi_s,s)ds\right|\xi_t=q\right]
\end{equation}
Hamilton's principle for the Guerra-Morato Lagrangian implies that
\begin{equation}
\left(\frac{\partial}{\partial t}+\sum_{j=1}^N b^j\nabla_j+\frac{1}{2}\Delta\right)S=\sum_{j=1}^N\left[\frac{1}{2}b^jb_j+\frac{1}{2}\nabla_jb^j-\Phi+A_jb^j+\frac{1}{2}\nabla_jb^j\right],
\end{equation}
which, since
\begin{equation}
b^j=\nabla^jS-A^j+\nabla^jR,
\end{equation}
becomes the Hamilton-Jacobi equation
\begin{equation}
\frac{\partial S}{\partial t}+\frac{1}{2}\sum_{j=1}^N(\nabla^jS-A^j)(\nabla_jS-A_j)+\Phi-\frac{1}{2}\sum_{j=1}^N\nabla^jR\nabla_jR-\frac{1}{2}\Delta R=0.
\end{equation}
The continuity equation
\begin{equation}
\frac{\partial \rho}{\partial t}=-\sum_{j=1}^N\nabla_j(v^j\rho),
\end{equation}
where
\begin{equation}
v^j=\frac{1}{2}(b^j+{b^j}^*)\qquad {b^j}^*=b^j-\nabla^j\log \rho,
\end{equation}
becomes
\begin{equation}
\frac{\partial R}{\partial t}+\sum_{j=1}^N(\nabla_jR)(\nabla^jS-A^j)+\frac{1}{2}\Delta S-\frac{1}{2}\nabla_jA^j=0.
\end{equation}
The non linear Hamilton-Jacobi and continuity PDE lead to the linear Schr\"odinger equation
\begin{equation}
i\frac{\partial \psi}{\partial t}= \underbrace{\left[\frac{1}{2}\sum_{j=1}^N\left(\frac{1}{i}\nabla^j-A^j\right)\left(\frac{1}{i}\nabla_j-A_j\right)+\Phi\right]}_{=:H}\psi,
\end{equation}
for the Schr\"odinger operator $H$, if we define the probability amplitude
\begin{equation}
\psi(q,t):=e^{R(q,t)+iS(q,t)}.
\end{equation}
Note that
\begin{equation}
\rho(q,t)=|\psi(q,t)|^2.
\end{equation}

\subsection{Solution to Schr\"odinger's Equation via Feynman's Path Integral}
The Hamilton function is the Legendre transformation of the Lagrangian:
\begin{equation}
H(p,q, t):=\left.\left(\sum_{j=1}^N p^jq^{\prime}_j-L(q,q^{\prime},t)\right)\right|_{p:=\frac{\partial L}{\partial q^{\prime}}=q^{\prime}+A}=\frac{1}{2}\sum_{j=1}^N\left(p^j-A^j\right)\left(p_j-A_j\right)+\Phi,
\end{equation}
and the Schr\"odinger operator is obtained by the quantization
\begin{equation}
q\rightarrow q\quad (\text{Multiplication operator})\qquad p\rightarrow \frac{1}{i}\nabla\quad (\text{Differential operator}).
\end{equation}
The solution of the Schr\"odinger initial value problem
\begin{equation}\label{SIVP}
\left\{
  \begin{array}{l}
    i\frac{\partial \psi}{\partial t}=H\psi\\
    \psi(q,0)=\psi_0(q),
  \end{array}
\right.
\end{equation}
can be obtained as the convolution of the initial condition with Feynman's path integral:
\begin{equation}
\psi(y,t)=\int\psi_o(q)\left(\int_{q(0)=q}^{q(t)=y}\exp\left(i\int_0^tL(u(s),u^{\prime}(s),s)ds\right) Du \right)dq,
\end{equation}
An approximation of Feynman's path integral can be obtained by averaging over a number of possible paths. If the original Lagrangian problem has to fulfill some constraints, these can be enforced in the choice of the paths to be averaged over in the integral.
\subsection{Application to Geometric Arbitrage Theory}
The geometric arbitrage theory Lagrangian reads
\begin{equation}
L(q,q^{\prime},t):= \frac{x\cdot(D^\prime+rD)}{x\cdot D},
\end{equation}
for $q:=(x,D,r)\in\mathbb{R}^{3N}$, where $x$, $D$ and $r$ represent portfolio nominals, deflators and short rates. The portfolios under consideration have to satisfy the self-financing condition
\begin{equation}
{x^{\prime}}\cdot D=0.
\end{equation}
Let us assume that the diffusions can be written separately as
\begin{equation}
\left\{
  \begin{array}{l}
    dx_t=b^x(t,x_t)dt+\sigma^x(x_t)dW_t \\
    dD_t=b^D(t,D_t)dt+\sigma^D(D_t)dW_t\\
    dr_t=b(t,r_t)dt+\sigma^r(r_t)dW_t,
  \end{array}
\right.
\end{equation}
\noindent where
\begin{equation}
\begin{split}
&b^x,b^D,b^r:[0,+\infty[\times\mathbb{R}^N\rightarrow\mathbb{R}^N\\
&\sigma^x, \sigma^D,\sigma^r:\mathbb{R}^N\rightarrow\mathbb{R}^{N\times K}
\end{split}
\end{equation}
\noindent are vector and matrix valued functions with appropriate regularity.\\
The geometric arbitrage theory Lagrangian can be written in the form (\ref{Lagr})
\begin{equation}
L(q,q^{\prime},t)=\left(\frac{1}{2}\sum_{j=1}^{3N}{q^{\prime}}^jq^{\prime}_j-\frac{1}{2}\right)+\Phi(q)+\sum_{j=1}^{3N}A_j(q){q^{\prime}}^j,
\end{equation}
\noindent if we set
\begin{equation}
\begin{aligned}
&\Phi(q):=-\frac{x\cdot(rD)}{x\cdot D}-\frac{1}{2}\qquad&A_j^D(q):=-\frac{{\sigma^D}^{-2}_{j,i}x_i}{x\cdot D}\\
&A_j^x(q):=0\qquad&A_j^r(q):=0.
\end{aligned}
\end{equation}
Therefore, the solution of Schr\"odinger's initial value problem (\ref{SIVP}) reads
\begin{equation}\label{Feynman_int}
\psi(y,t)=\int\psi_o(q)\left(\int_{q(0)=q}^{q(t)=y}\exp\left(i\int_0^t \frac{x_s\cdot(D_s^{\prime}+r_sD_s)}{x_s\cdot D_s}ds\right) Du \right)dq,
\end{equation}
where the Feynman integration is over all paths satisfying the constraint
\begin{equation}\label{constraint_int}
   {x^{\prime}}\cdot D \equiv0.
\end{equation}
The first constraints is satisfied by all path where time is the arc length parameter. Formula (\ref{Feynman_int}) can be approximatively computed by Monte-Carlo methods simulating system trajectories satisfying the constraints (\ref{constraint_int}), leading to numerical efficient computations.

\section{Conclusion}
By introducing an appropriate stochastic differential geometric
formalism, the classical theory of stochastic finance can be
embedded into a conceptual framework called geometric
arbitrage theory, where the market is modelled with a principal
fibre bundle with a connection whose curvature corresponds to the instantaneous arbitrage capability.
The market and its dynamic can be seen as a  stochastic Lagrangian system, or, equivalently, as a stochastic Hamiltonian system.
Instead of trying to compute direct a solution of the stochastic Hamilton equations, we quantize a deterministic version of the Hamiltonian system.
The asset and market dynamics have then a quantum mechanical formulation in terms of Schr\"odinger equation which can be solved numerically by means of Feynman's integrals. Ehrenfest's theorem and Heisenberg's uncertainty relation lead to new econometric results: in the equilibrium minimal arbitrage returns on asset values and nominal values are centered and serially uncorrelated; the variances of an asset value and its corresponding weight in the market portfolio cannot be both arbitrarily small: if one increases,  the other decreases and viceversa.


\appendix

\section{Generalized Derivatives of Stochastic
Processes}\label{Derivatives} In stochastic differential geometry one would like to lift
the constructions of stochastic analysis from open subsets of
$\mathbf{R}^N$ to  $N$ dimensional differentiable manifolds. To that
aim, chart invariant definitions are needed and hence a stochastic
calculus satisfying the usual chain rule and not It\^{o}'s Lemma is required,
(cf. \cite{HaTh94}, Chapter 7, and the remark in Chapter 4 at the
beginning of page 200). That is why the papers about geometric arbitrage theory are mainly concerned in
 by stochastic integrals and derivatives meant in \textit{Stratonovich}'s
sense and not in \textit{It\^{o}}'s. Of course, at the end of the computation, Stratonovich integrals can be transformed into It\^{o}'s.
Note that a fundamental portfolio equation, the self-financing condition cannot be directly formally expressed with Stratonovich integrals, but first with It\^{o}'s and then transformed into Stratonovich's, because it is a non-anticipative condition.
\begin{defi}\label{Nelson}
Let $I$ be a real interval and $Q=(Q_t)_{t\in I}$ be a  $\mathbb{R}^N$-valued stochastic process on the probability space
$(\Omega, \mathcal{A}, P)$. The process $Q$ determines three families of $\sigma$-subalgebras of the $\sigma$-algebra $\mathcal{A}$:
\begin{itemize}
\item[(i)] ''Past'' $\mathcal{P}_t$, generated by the preimages of Borel sets in $\mathbf{R}^N$  by all mappings $Q_s:\Omega\rightarrow\mathbf{R}^N$ for $0<s<t$.
\item[(ii)] ''Future'' $\mathcal{F}_t$, generated by the preimages of Borel sets in $\mathbf{R}^N$  by all mappings $Q_s:\Omega\rightarrow\mathbf{R}^N$ for $0<t<s$.
\item[(iii)] ''Present'' $\mathcal{N}_t$, generated by the preimages of Borel sets in $\mathbf{R}^N$  by the mapping $Q_s:\Omega\rightarrow\mathbf{R}^N$.
\end{itemize}
Let $Q=(Q_t)_{t\in I}$ be continuous.
 Assuming that the following limits exist,
\textbf{Nelson's stochastic derivatives} are defined as
\begin{equation}
\boxed{
\begin{split}
&\mathfrak{D}Q_t:=\lim_{h\rightarrow
0^+}\mathbb{E}\Br{\left.\frac{Q_{t+h}-Q_t}{h}\right|
\mathcal{P}_t}\text{: forward derivative,}\\
& \mathfrak{D}_*Q_t:=\lim_{h\rightarrow
0^+}\mathbb{E}\Br{\left.\frac{Q_{t}-Q_{t-h}}{h}\right|
\mathcal{F}_{t}}\text{: backward derivative,}\\
&\mathcal{D}Q_t:=\frac{\mathfrak{D}Q_t+\mathfrak{D}_*Q_t}{2}\text{: mean derivative}.
\end{split}
}
\end{equation}
Let $\mathcal{S}^1(I)$ the set of all processes $Q$ such that
$t\mapsto Q_t$, $t\mapsto \mathfrak{D}Q_t$ and $t\mapsto \mathfrak{D}_*Q_t$ are continuous
mappings from $I$ to $L^2(\Omega, \mathcal{A})$. Let
$\mathcal{C}^1(I)$ the completion of $\mathcal{S}^1(I)$ with respect
to the norm
\begin{equation}
\boxed{
\|Q\|:=\sup_{t\in I}\br{\|Q_t\|_{L^2(\Omega, \mathcal{A})}+\|\mathfrak{D}Q_t\|_{L^2(\Omega, \mathcal{A})}+\|\mathfrak{D}_*Q_t\|_{L^2(\Omega, \mathcal{A})}}.
}
\end{equation}
\end{defi}

\begin{remark}
The stochastic derivatives $\mathfrak{D}$, $\mathfrak{D}_*$ and  $\mathcal{D}$
correspond to It\^{o}'s, to the anticipative and, respectively,  to Stratonovich's integral (cf. \cite{Gl11}). The process space $\mathcal{C}^1(I)$ contains all It\^{o} processes. If $Q$ is a Markov process, then the sigma algebras $\mathcal{P}_t$ (''past'') and $\mathcal{F}_t$ (''future'') in the definitions of forward and backward derivatives can be substituted by the sigma algebra $\mathcal{N}_t$ (''present''), see Chapter 6.1 and 8.1 in (\cite{Gl11}).
\end{remark}

\noindent Stochastic derivatives can be defined pointwise in $\omega\in\Omega$ outside the class $\mathcal{C}^1$ in terms of generalized functions.
\begin{defi}
Let $Q:I\times\Omega\rightarrow\mathbb{R}^N$ be a continuous linear functional in the test processes $\varphi:I\times\Omega\rightarrow\mathbb{R}^N$ for $\varphi(\cdot,\omega)\in C^{\infty}_c(I,\mathbb{R}^N)$. We mean by this that for a fixed $\omega\in\Omega$ the functional $Q(\cdot,\omega)\in\mathcal{D}(I,\mathbb{R}^N)$, the topological vector space of continuous distributions. We can then define
\textbf{Nelson's generalized stochastic derivatives:}
\begin{equation}
\boxed{
\begin{split}
&\mathfrak{D}Q(\varphi_t):=-Q(\mathfrak{D}\varphi_t)\text{: forward generalized derivative,}\\
& \mathfrak{D}_*Q(\varphi_t):=-Q(\mathfrak{D}_*\varphi_t)\text{: backward generalized derivative,}\\
&\mathcal{D}(\varphi_t):=-Q(\mathcal{D}\varphi_t)\text{: mean generalized derivative}.
\end{split}
}
\end{equation}
\end{defi}
\noindent If the generalized derivative is regular, then the process has a derivative in the classic sense. This construction is nothing else than a straightforward pathwise lift of the theory of generalized functions to a wider class stochastic processes which do not a priori allow for Nelson's derivatives in the strong sense. We will utilize this feature in the treatment of credit risk, where many processes with jumps occur.

\end{document}